\documentclass{amsart}
\usepackage{amssymb, latexsym}

\newcommand{\beqa}{\begin{eqnarray*}}
\newcommand{\eeqa}{\end{eqnarray*}}
\newcommand{\beqn}{\begin{eqnarray}}
\newcommand{\eeqn}{\end{eqnarray}}

\newcommand{\iy}{\infty}
\newcommand{\lt}{\left}
\newcommand{\rt}{\right}

\newcommand{\R}{\mathbb R}
\newcommand{\Pa}{\mathbb P}
\newcommand{\bQ}{\mathbb Q}

\newcommand{\N}{\mathbb N}
\newcommand{\Ha}{\mathbb H}
\newcommand{\D}{\mathbb D}
\newcommand{\B}{\mathbb B}
\newcommand{\C}{\mathbb C}
\newcommand{\mcE}{\mathcal E}

\newcommand{\mcB}{\mathcal B}
\newcommand{\mcA}{\mathcal A}

\newcommand{\f}{\frac}
\newcommand{\tf}{\tfrac}

\newcommand{\al}{\alpha}
\newcommand{\be}{\beta}
\newcommand{\G}{\Gamma}
\newcommand{\g}{\gamma}
\newcommand{\e}{\varepsilon}

\newcommand{\de}{\delta}

\newcommand{\la}{\lambda}

\newcommand{\Om}{ \Omega}

\newcommand{\s}{\sigma}
\newcounter{cnt1}
\newcounter{cnt2}
\newcounter{cnt3}
\newcommand{\blr}{\begin{list}{$($\roman{cnt1}$)$}
 {\usecounter{cnt1} \setlength{\topsep}{0pt}
 \setlength{\itemsep}{0pt}}}
\newcommand{\bla}{\begin{list}{$($\alph{cnt2}$)$}
 {\usecounter{cnt2} \setlength{\topsep}{0pt}
 \setlength{\itemsep}{0pt}}}
\newcommand{\bln}{\begin{list}{$($\arabic{cnt3}$)$}
 {\usecounter{cnt3} \setlength{\topsep}{0pt}
 \setlength{\itemsep}{0pt}}}
\newcommand{\el}{\end{list}}
\newtheorem{thm}{Theorem}
\newtheorem{lem}[thm]{Lemma}
\newtheorem{cor}[thm]{Corollary}

\newtheorem{Def}[thm]{Definition}

\newtheorem{rem}[thm]{Remark}
\newcommand{\Rem}{\begin{rem} \rm}
\newcommand{\bdfn}{\begin{Def} \rm}
\newcommand{\edfn}{\end{Def}}

\newcommand{\ba}{\begin{array}}
\newcommand{\ea}{\end{array}}

\date{}
\begin{document}
\title{\bf Global solutions to the homogeneous and inhomogeneous Navier-Stokes equations}
\author[Gill]{T. L. Gill}
\address[Tepper L. Gill]{Department of Mathematics, E\&CE and Computational Physics Laboratory, Howard University\\
Washington DC 20059 \\ USA, {\it E-mail~:} {\tt tgill@howard.edu}}
\author[Williams]{D. Williams}
\address[Daniel Williams]{ Department of Mathematics, Howard University \\
Washington DC 20059 \\ USA, {\it E-mail~:} {\tt dwilliams@howard.edu}}
\author[Zachary]{W. W. Zachary*}
\address[Woodford W. Zachary]{Department of Mathematics, E\&CE and Computational Physics Laboratory, \\ Howard
University\\ Washington DC 20059 \\ USA, {\it E-mail~:} {\tt
wwzachary@earthlink.net}}
\thanks{*deceased}
\date{}
\begin{abstract}  In this paper we take a new approach to a proof of existence and uniqueness of solutions for the 3D-Navier-Stokes equations, which leads to essentially the same proof for both bounded and unbounded domains and for homogeneous or inhomogeneous incompressible fluids.    Our approach is to construct the largest separable Hilbert space ${\bf{SD}}^2[\R^3]$, for which the Leray-Hopf (type) solutions in $L^2[{\mathbb R}^3]$ are strong solutions in ${\bf{SD}}^2[\R^3]$.  We say Leray-Hopf type because our solutions are weak in the spatial sense but not in time.

When the body force is zero, we prove that, there exists a positive constant ${{{u}}_ +}$, such that, for all divergence-free vector fields in a dense set $\mathbb{D}$ contained in the closed ball ${{\mathbb B}}$ of radius $\tfrac{(1-\varepsilon)}{2}{{u}_ +}, \ 0< \varepsilon <1$, the initial value problem has unique global weak solutions in ${\mathbb C}^{1} \left( {(0,\infty ),{{\mathbb B}}} \right)$.   When the body force is nonzero, we obtain the same result for vector fields in a dense set $\mathbb{D}$ contained in the annulus bounded by constants ${u}_{-}$ and $\tfrac{1}{2}{u}_ {+}$.   In either case, we obtain existence and uniqueness for the Leray-Hopf weak solutions on ${\mathbb R}^3$.  Moreover, with mild conditions on the decay properties of the initial data, we obtain pointwise and time-decay of the solutions.  
\end{abstract}
\maketitle
\section*{Introduction} 
Let ${[L^2({\mathbb{R}}^3)]^3}$ be the Hilbert space of square integrable functions on ${\mathbb {R}}^3$,  let ${\mathbb {H}}[ {\mathbb {R}}^3 ]$ be the completion of the  set of functions in $\left\{ {{\bf{u}} \in \mathbb {C}_0^\infty  [ {\mathbb {R}}^3 ]^3 \left. {} \right|\,\nabla  \cdot {\bf{u}} = 0} \right\}$ which vanish at infinity with respect to the inner product of ${[L^2({\mathbb{R}}^3)]^3 }$, and let ${\mathbb{V}}[ {\mathbb {R}}^3 ]$ be the completion of the above functions which vanish at infinity with respect to the inner product of $\mathbb{H}^1[ {\mathbb {R}}^3 ]$, the functions in ${\mathbb{H}} [ {\mathbb {R}}^3 ]$ with weak derivatives in ${[L^2({\mathbb{R}}^3)]^3 }$.  The classical Navier-Stokes initial-value problem (on $ \mathbb{R}^3 {\text{ and all }}T > 0$) is to find a  function ${\mathbf{u}}:[0,T] \times {\mathbb {R}}^3  \to \mathbb{R}^3$ and $p:[0,T] \times {\mathbb {R}}^3  \to \mathbb{R}$ such that
\beqn
\begin{gathered}
  \partial _t {\mathbf{u}} + ({\mathbf{u}} \cdot \nabla ){\mathbf{u}} - \nu \Delta {\mathbf{u}} + \nabla p = {\mathbf{f}}(t){\text{ in (}}0,T) \times {\mathbb {R}}^3 , \hfill \\
  {\text{                              }}\nabla  \cdot {\mathbf{u}} = 0{\text{ in (}}0,T) \times {\mathbb {R}}^3 {\text{ (in the weak sense),}} \hfill \\
    {\text{                              }}{\mathbf{u}}(0,{\mathbf{x}}) = {\mathbf{u}}_0 ({\mathbf{x}}){\text{ in }}{\mathbb {R}}^3. \hfill \\ 
\end{gathered} 
\eeqn
The equations describe the time evolution of the fluid velocity ${\mathbf{u}}({\mathbf{x}},t)$ and the pressure $p$ of an incompressible viscous homogeneous Newtonian fluid with constant viscosity coefficient $\nu $ in terms of a given initial velocity ${\mathbf{u}}_0 ({\mathbf{x}})$ and given external body forces ${\mathbf{f}}({\mathbf{x}},t)$.  

Let $\mathbb{P}$ be the (Leray) orthogonal projection of 
$(L^2 [ {\mathbb {R}}^3 ])^3$ 
onto ${{\mathbb{H}}}[ {\mathbb {R}}^3]$ and define the Stokes operator by:  $ {\bf{Au}} = : -\mathbb{P} \Delta {\bf{u}}$, 
for ${\bf{u}} \in D({\bf{A}}) \subset {\mathbb{H}}^{2}[ {\mathbb {R}}^3]$, the domain of ${\bf{A}}$.      If we apply $\mathbb{P}$ to equation (1), with 
${{B}}({\mathbf{u}},{\mathbf{u}}) = \mathbb{P}({\mathbf{u}} \cdot \nabla ){\mathbf{u}}$, we can recast equation (1) into the standard form:
\beqn
\begin{gathered}
  \partial _t {\mathbf{u}} =  - \nu {\mathbf{Au}} - {{B}}({\mathbf{u}},{\mathbf{u}}) + \mathbb{P}{\mathbf{f}}(t){\text{ in (}}0,T) \times \R^3 , \hfill \\
  {\text{                              }}{\mathbf{u}}(0,{\mathbf{x}}) = {\mathbf{u}}_0 ({\mathbf{x}}){\text{ in }}\R^3, \hfill \\ 
\end{gathered} 
\eeqn
where the orthogonal complement of ${\Ha} $ relative to $\{{L}^{2}(\R^3)\}^3, \;  \{ {\mathbf{v}}\,:\;{\mathbf{v}} = \nabla q,\;q \in \Ha^1[\R^3] \}$, is used to eliminate the pressure term (see Galdi [GA] or [SY, T1,T2]). 
\subsection*{Background}
The existence of global weak solutions for (2) was proved by Leray \cite{Le} for all divergence-free initial data ${\bf{u}}_0 \in {{\mathbb H}( {\mathbb {R}}^3 )}$.  (Hopf  \cite{Ho} solved the same problem for a bounded open domain $ \Om  \subset \mathbb{R}^n , n \ge 2$ (see also \cite{Li1, T1, vW}).)

Leray used
\beqn
{\mathbf{b}}\left( {{\mathbf{u}},{\mathbf{u}},{\mathbf{u}}} \right) = \left\langle {{{B}}\left( {{\mathbf{u}},{\mathbf{u}}} \right),{\mathbf{u}}} \right\rangle _\Ha  = \int_{\R^3} {\left[ {{\mathbf{u}}({\mathbf{x}}) \cdot \nabla {\mathbf{u}}({\mathbf{x}})} \right] \cdot {\mathbf{u}}({\mathbf{x}})d{\mathbf{x}} = 0} 
\eeqn
to show that, for such initial data, the global solution  ${\bf u}(t,{\bf x})$ satisfies the well-known energy inequality:
\[
\left\| {{\mathbf{u}}(t)} \right\|_\mathbb{H}^2  + 2\nu \int_0^t {\left\| {{\mathbf{A}}^{1/2} {\mathbf{u}}(s)} \right\|_\mathbb{H}^2 ds}  \leqslant \left\| {{\mathbf{u}}_0 } \right\|_\mathbb{H}^2 ,\quad \forall t \geqslant 0.
\]
These solutions are called Leray-Hopf solutions. There are two open
questions in this case.  The first is whether or not all Leray-Hopf solutions are unique and the second is whether or not those solutions with smooth initial data are regular. (A weak solution ${\bf u}(t, {\bf x})$ is regular if $\left\| {{\mathbf{u}}(t)} \right\|_{\mathbb{V} }$ is continuous.)

Until $1964$, another open problem was the existence of  global-in-time strong solutions (in the $\Ha$ norm) for the three-dimensional Navier-Stokes initial value problem.  In that year, Fujita and Kato \cite{FK} proved that  strong, global-in-time, smooth three-dimensional solutions exist in the Sobolev space ${\Ha}^{1/2}$ provided that the body forces are small and the initial data is small  (compared to the viscosity term $-\nu{\bf A}{\bf u}$, see Section 3 and also \cite{KF}, \cite{CH} and \cite{T3}).  Their work was extended to $L^p$ spaces by F. Weissler \cite{WE} and considered in Besov spaces (of negative index of regularity) by M. Cannone, Y. Meyer and F. Planchon \cite{CMP}.

Since then, a number of papers have appeared proving existence of global-in-time, solutions for small initial data, which are strong in a particular norm.  These results will be discussed briefly in Section 3 but, the interested reader is directed to \cite{LE} (see also \cite{KA1}, \cite{FRT}, \cite{CA}, \cite{PL} and \cite{KT}).  The authors of \cite{GIP} observe that, in Leray's theory the nonlinear term becomes zero because of equation (3).  They note that, ``For global existence to hold, the point in those ``strong solution'' theorems is that the smallest assumption enables one to get rid of the nonlinear term, which can be absorbed by the Laplacian.''   

The  interesting paper of Chemin and Gallagher \cite{CG} discusses all the relevant spaces, provides their own approach to the problem, and gives a nice picture of the kind of results one can expect from the methods used.  Many of the recent approaches exploit the interesting  invariance properties of the Navier-Stokes equations (with zero body forces) to construct their spaces.  However, it appears that these spaces do not maintain their invariance properties when body forces are present.  Furthermore, as first observed by Kato \cite{KA1},  strong global solutions in $L^n[\R^n]$ for example,  are not necessarily weak solutions in the sense of Leray-Hopf (i.e., they need not have finite energy norm). 
\subsection*{Purpose}
Our approach differs sharply from other attempts.  We first construct the largest separable Hilbert space, for which the Leray-Hopf (type) solutions on ${\mathbb R}^3$ are strong ones.  We then obtain the smallest viscosity and largest body forces that balance each other in such a way as to allow global strong solutions for reasonable velocities (see below). A major advantage is that, the methods developed also apply to bounded domains and inhomogeneous fluids (with minor adjustments).     
\subsection*{Asymptotic Properties}
A number of studies have been conducted on the asymptotic and stability behavior of solutions, ${\bf u}(t, {\bf x})$, of the Navier-Stokes equations.  The paper by  Brandolese and Vigneron \cite{BV} provides a comprehensive analysis of the problem and a clear presentation of the latest results in this direction (see also \cite{GIP} and Miyakawa \cite{MI}).  
\subsubsection*{The Problem}
The general problem in a bounded and unbounded domain is closely related. However, there is one major difference in the two cases.  In order to understand the nature of an additional difficulty  for the unbounded domain, it is important to discuss a problem that occurs when the domain $\Om \subset \R^3$ is bounded.  

In the bounded domain case, it has been shown by Bru\u{s}linskaja \cite{BR} that,  if the viscosity coefficient $\nu$ is sufficiently small, then stationary solutions of (2)  lose stability and at least one eigenvalue of the linear Stokes operator passes from the left halfplane to the right halfplane.  This means that the stationary solution becomes a limit cycle, which may be either stable or unstable. In either case, these bifurcations result in the existence of nonunique solutions  of the three-dimensional Navier-Stokes equations. This problem has been discussed in a more general way by Foias and Temam \cite{FT}.   From this, it's clear that the viscosity coefficient is not a passive constant, but plays an important role in determining the physical properties of the solution(s).  In fact, if $\la_1$ is the first eigenvalue of the Stokes operator, the quantity
\[
G= \frac{{\left\| {\mathbf{f}} \right\|_\infty ^2 }}
{{\nu ^2 \lambda _1^{3/4} }},
\]
known as the Grashof number, appears naturally and provides a measure of the dynamical complexity of solutions.  The Grashof number is similar to the Reynolds number and the dynamical complexity increases with increasing G (see Foias et al \cite{FMTT}).   For both physical and mathematical reasons, the corresponding unbounded domain problem is  more difficult to study and, to our knowledge, has not received attention in the literature. 
\subsubsection{Stability of flow}    
The natural requirement of stability of flow in $\R^3$ is usually implemented with the requirement that, under reasonable conditions, the velocity vector field ${\bf{u}}(t)$ approaches zero for large $t$.  However, this implies that the total energy of the fluid also approaches zero.  Physically this means that the boundary at infinity is kept (at least) at the solid state phase point (i.e., zero degrees for water). Since the physical properties of our fluid change radically at this point, we must be slightly more precise in this case. 

On the other hand, in turbulent flow, there is a very important difference between the two and three-dimensional case. In the two-dimensional case, the fluid kinetic energy is transferred from both large to small and small to large scales by nonlinear interactions between different scales of motion (see Fj\"{o}rtoft \cite{FJ} and Thompson \cite{TH}). However, in the three-dimensional case, there is a one way nonlinear cascade from large to small scales of motion. It follows that, as the kinetic energy becomes large, the fluid velocity becomes more erratic at smaller and smaller scales. It is generally assumed that the condition for a nonturbulent flow is captured by the size of the viscosity coefficient. Physically, this is a good measure, but also implies a number of other well-defined conditions, usually related to body forces, temperature, pressure and relative domain configuration (i.e., obstacles, constrictions, etc).  Thus, any reasonable solution should also lead to an upper bound on the total energy (i.e., the velocity in ${\bf L}^2$-norm). 
The current discussion implies some bounds on the body forces.  We will be more precise later (see Theorem 25) but for now, it suffices to assume that $f=\sup _{t \in {\mathbf{R}}^ +  } \left\| {\mathbb{P}{\mathbf{f}}(t)} \right\|_{{\mathbb{H}}}  < \infty $.
\begin{Def}  We say that a velocity vector field in $\R^3$ is \underline{\rm reasonable} if for  $0 \le t<\iy$, there is a continuous function $m(t)>0$, depending only on $t$ and a constant $M_0$, which may depend on ${\bf u}(0)$ and $f$, such that 
\[
0 < m(t) \leqslant \left\| { {\mathbf{u}}(t)} \right\|_{{\mathbb{H}}} \le M_0.
\] 
\end{Def}
The above definition formalizes the requirement that the fluid has bounded positive definite energy,  However, this condition still allows the velocity to approach zero at infinity in a weaker norm.
\subsection{Statement of Results}
Let ${\bf SD}^2[\R^3]$ be our separable Hilbert space, which is constructed in Section 1.1 and contains ${[L^2({\mathbb{R}}^3)]^3}$ as a compact dense embedding.  Let ${\mathbb {H}}_{sd}$ be the completion in the ${\bf{SD}}^2[ {\mathbb {R}}^3]$ norm of the set of functions in $\left\{ {{\bf{u}} \in \C_0^{\iy}[ {\mathbb {R}}^3]^3 \left. {} \right|\,\nabla  \cdot {\bf{u}} = 0} \right\}$ and let  $\mathbb{P}$ be the  orthogonal projection of ${\bf{SD}}^2[ {\mathbb {R}}^3]$ onto ${{\mathbb{H}}_{sd}}$.

Rewrite the first  equation in (2) in the form:
\beqn
\begin{gathered}
  \partial _t {\mathbf{u}} =  {\mathcal{A}} ({\mathbf{u}},t) {\text{ in (}}0,T) \times \R^3 , \hfill \\
  {\mathcal{A}}({\mathbf{u}},t) =  - \nu{\bf A}{\mathbf{u}} -  {{B}}({\mathbf{u}},{\mathbf{u}}) +  \mathbb{P}{\mathbf{f}}(t). \hfill \\ 
\end{gathered} 
\eeqn
Let $\B$ be a closed convex subset of $\Ha_{sd}$, which will be identified during the proof of the following theorem.
\begin{thm} If ${\bf f} \ne 0$, for each $t \in [0, \iy)$ there exist positive constants ${{u}}_ +, \; {u}_-  $, depending on $f$ and $\nu $  such that, for all initial data ${\mathbf{u}}_0 \in \B \cap D({\bf A}) \subset \Ha_{sd}$ with $0 \le {u}_- < \left\| {\mathbf{u}}_0 \right\|_{sd}  \le \tf{1}{2}{{u}}_ +$, the operator ${\mathcal{A}}( \cdot ,t)$ is the generator of a strongly continuous nonlinear contraction semigroup on $\mathbb{B}$. If ${\bf f}=0$, we replace $\tf{1}{2}u_+$ by $\tf{(1-\e)}{2}u_+$, where $0<\e<1$, and use the ball $\B$  of radius $\tf{(1-\e)}{2}u_+$ centered at the origin.
\end{thm}
If $T(t)$ is the nonlinear semigroup generated by ${\mcA}( \cdot ,t)$, then ${\bf u}(t, {\bf x})=T(t){\bf u}_0({\bf x})$ solves the initial value problem (2).  
We now have: 
\begin{thm} For each $T \in {\mathbf{R}}^ +$, $t \in (0,T)$ and ${\mathbf{u}}_0  \in \mathbb{B} \cap D({\bf A})$, the global-in-time Navier-Stokes initial-value problem in $\mathbb{R}^3 :$
\beqn
\begin{gathered}
  \partial _t {\mathbf{u}} + ({\mathbf{u}} \cdot \nabla ){\mathbf{u}} - \nu \Delta {\mathbf{u}} + \nabla p = {\mathbf{0}}{\text{ in (}}0,T) \times \mathbb{R}^3  , \hfill \\
  {\text{                              }}\nabla  \cdot {\mathbf{u}} = 0{\text{ in (}}0,T) \times \mathbb{R}^3  , \hfill \\
   {\text{                              }}{\mathbf{u}}(0,{\mathbf{x}}) = {\mathbf{u}}_0 ({\mathbf{x}}){\text{ in }}\mathbb{R}^3,  \hfill \\ 
\end{gathered} 
\eeqn
 has a unique strong solution ${\mathbf{u}}(t,{\mathbf{x}})$, which is in
 ${\bf SD}^2[[0,\infty); {\mathbb {H}}_{sd}]$.
\end{thm}
\begin{thm} If $\mathbf{u}(t, {\bf x})$ is a solution in the sense of Kato \cite{KA1}, in any one of the ${\bf L}^p[\R^3]$ spaces, or any space $\mcB$, which is continuously embedded in ${\bf L}^p[\R^3]$,  then $\mathbf{u}(t, {\bf x})$ is a solution in ${\bf SD}^2[\R^3]$.
\end{thm}
Let $S(t)$ be the semigroup generated by the Stokes operator. It is well-known that $\mathop {\lim }\nolimits_{t \to \infty }S(t){\mathbf{u}}_0 =0$. Since any strong solution is a mild solution, we also have that
\[
{\mathbf{u}}(t,{\mathbf{x}}) = S(t){\mathbf{u}}_0 ({\mathbf{x}}) + \int_0^t {S(t - s)\left[- {{{B}}\left( {{\mathbf{u}}(s,{\mathbf{x}}),{\mathbf{u}}(s,{\mathbf{x}})} \right) + \mathbb{P}{\mathbf{f}}(s)} \right]ds}. 
\]
If we introduce the following energy matrices due to Brandolese and Vigneron,
\[
\mathcal{E}_{h,k} (t) = \int_{\mathbb{R}^3 } {(u_h u_k )({\mathbf{x}},t)d{\mathbf{x}}} \quad {\text{and}}\quad {\mathcal{K}}_{h,k} (t) = \int_0^t {\int_{\mathbb{R}^3 } {(u_h u_k )({\mathbf{x}},s)d{\mathbf{x}}ds} }, 
\]
then we have:
\begin{thm} Let ${\bf u}(t, {\bf x})=T(t){\mathbf{u}}_0({\mathbf{x}})$ be the solution to the Navier-Stokes initial value problem (2).  If ${\bf f}(t)={\bf 0}$ and there is a $\de>0$ such that
\beqn
ess\mathop {\sup }\limits_{{\mathbf{x}} \in \mathbb{R}^3 } \left( {1 + \left| {\mathbf{x}} \right|} \right)^{2 + \de } \left| {\bf A}{{\mathbf{u}}_0({\mathbf{x}})} \right| < \infty,
\eeqn
then 
\begin{enumerate}
\item There is a constant $\g$ such that
\[
{\mathbf{u}}(t,{\mathbf{x}}) = S(t){\mathbf{u}}_0 ({\mathbf{x}}) + \gamma \nabla \left( {\sum\nolimits_{h,k} {\frac{{\delta _{h,k} \left| {\mathbf{x}} \right|^2  - 3x_h x_k }}
{{3\left| {\mathbf{x}} \right|^5 }}{\text{K}}_{h,k} (t)} } \right) + 0\left( {\frac{1}
{{\left| {\mathbf{x}} \right|^4 }}} \right).
\]
\item  There is a constant $p_0$ such that
\[
p(t,{\mathbf{x}}) = p_0  - \gamma \sum\nolimits_{h,k} {\left( {\frac{{\delta _{h,k} }}
{{3\left| {\mathbf{x}} \right|^3 }} - \frac{{x_h x_k }}
{{\left| {\mathbf{x}} \right|^5 }}} \right)} {\text{E}}_{h,k} (t) + O_t \left( {\frac{1}
{{\left| {\mathbf{x}} \right|^4 }}} \right).
\]
\end{enumerate}
\end{thm} 
It follows from this that, assuming Theorem 4,  equation (6) is sufficient to ensure stability of global solutions to the Navier-Stokes equations. Furthermore, Picard's iterative scheme clearly applies in this case.

It is  known that, if ${\mathbf{u}_0} \in \mathbb{V}$ and $\mathbf{f}(t) \in L^{\infty}[(0,\infty), \mathbb{H}]$, then there is a time $T> 0$ such that a Leray-Hopf (weak) solution of the Navier-Stokes equation  is uniquely determined on any subinterval of $[0,T)$ (see Sell and You, \cite{SY} p. 396).   Thus, we also have that: 
\begin{cor} For each $t \in {\mathbf{R}}^ + $ and $
{\mathbf{u}}_0  \in \B \cap D({\bf A})$, the Navier-Stokes initial-value problem in $ \mathbb{R}^3 :$
\beqn
\begin{gathered}
  \partial _t {\mathbf{u}} + ({\mathbf{u}} \cdot \nabla ){\mathbf{u}} - \nu \Delta {\mathbf{u}} + \nabla p = {\mathbf{f}}(t){\text{ in (}}0,T) \times \R^3 , \hfill \\
  {\text{                              }}\nabla  \cdot {\mathbf{u}} = 0{\text{ in (}}0,T) \times \R^3 , \hfill \\
  {\text{                              }}{\mathbf{u}}(0,{\mathbf{x}}) = {\bf u}_0({\bf x}) , \hfill \\
\end{gathered} 
\eeqn
 has a unique weak solution
${\mathbf{u}}(t,{\mathbf{x}})$, which is in
 ${L_{\text{loc}}^2}[[0,\infty); {\mathbb {H}}_{sd}]$ and in
$L_{\text{loc}}^\infty[[0,\infty); {\mathbb V}_{sd}]
\cap \mathbb{C}^1[(0,\infty);{\mathbb H}_{sd}]$.  Moreover, in this case, we also have that $\mathop {\lim }\limits_{t \to \infty } \left\| {{\mathbf{u}}(t)} \right\|_{\mathbb{H}_{sd} }=0$.  To be more precise, 
\[
\mathop {\lim }\limits_{t \to \infty } \left\| {{\mathbf{u}}(t) - S(t){\mathbf{u}}_0 } \right\|_{\mathbb{H}_{sd} }  = O(t^{ - \alpha /2} ),
\]
where $0< \al < 1/2$.
\end{cor}
It follows from here that existence in the sense of Leray-Hopf is suffcient to ensure asymptotic decay of global solutions to the Navier-Stokes equations.
\subsection{Summary}  In the first section, we bring together a number of basic analytic tools that we use to prove our main results.   We then construct our Hilbert space and obtain  strong a priori bounds for the nonlinear term in the Navier-Stokes equations.  The second section is devoted to proofs of our main results.  In the third section, we discuss how our approach  allows us to solve the inhomogeneous problem (on $\R^3$).  Finally, we note that, with minor changes, our results also apply to the bounded domain case for both homogeneous and inhomogeneous fluids.
\section{Basic Tools}
We now establish a number of results that will be used in the sequel.
\begin{Def}  We say that a (generally nonlinear) operator ${\mathcal{A}}( \cdot ,t)$ is (for each $t$) 
\begin{enumerate}
\item
0-Dissipative if $
\left\langle {{\mathcal{A}}({\mathbf{u}},t),{\mathbf{u}}} \right\rangle _{{\mathbb{H}}_{sd}}  \le 0$,
\item
Dissipative if 
$\left\langle {{\mathcal{A}}({\mathbf{u}},t) - {\mathcal{A}}({\mathbf{v}},t),{\mathbf{u}} - {\mathbf{v}}} \right\rangle _{{\mathbb{H}}_{sd}}  \le 0$,
\item
Strongly dissipative if there  exists a $ \be > 0$, which may depend on $t$, such that
\[
\left\langle {{\mathcal{A}}({\mathbf{u}},t) - {\mathcal{A}}({\mathbf{v}},t),{\mathbf{u}} - {\mathbf{v}}} \right\rangle_{{\mathbb{H}}_{sd}}  \le  - \be \left\| {{\mathbf{u}} - {\mathbf{v}}} \right\|_{{\mathbb{H}}_{sd}}^2. 
\]
\end{enumerate}
\end{Def}

Note that, if ${\mathcal{A}}( \cdot ,t)$ is a linear operator, definitions (1) and (2) coincide.  Theorem 9 below is essentially due to Browder \cite{B},   while Theorem 10 is  a slight extension of one from Miyadera \cite[p. 185, Theorem 6.20]{M}.  The extension follows from Theorem A1 and Theorem A2 of Crandall and Pazy \cite{CP} along with the time-dependent version of the Crandall-Liggett Theorem \cite{CL} (see the appendix to the first section of \cite{CL}). Taken together, this is an  extension of Theorems I and II in Kato \cite{KA2}.  
\begin{thm} Let $\mathbb{B}$ be a closed, bounded, convex subset of $
{{\mathbb{H}}_{sd}}$.  If ${\mathcal{A}}( \cdot ,t):D({{\mathcal{A}}( \cdot ,t)}) \cap \mathbb{B} \to {{\mathbb{H}}_{sd}}$ is  a densely defined strongly dissipative mapping for each fixed $t \in [0, \iy)$, then for each $\la>0$, $
Ran{\text{[}}{I-\la\mathcal{A}}( \cdot ,t)] \supset \mathbb{B}$).
\end{thm}	
\begin{thm} Let  $\mathbb{B}$ is a closed convex set and let ${\mathcal{A}( \cdot ,t)}, t \in I = [0,\infty )$ be a  densely defined family of operators  on ${{\mathbb{H}}_{sd}}$ with domains $D({{\mathcal{A}}( \cdot ,t)}) \cap \mathbb{B}) = \D$, independent of $t$, such that:
\begin{enumerate}
\item
The operator $\mathcal{A}( \cdot ,t)$ is the generator of a contraction semigroup on $\D$ for each
$t \in I$.
\item
The function $\mathcal{A}({\mathbf{u}}, t)$ is continuous in both variables on $
 \D \times I $.
\end{enumerate}
Then $\mathcal{A}( \cdot ,t)$ extends uniquely to the generator of a contraction semigroup on $\B$ and, for every ${\mathbf{u}}_0  \in \D \cap \mathbb{B}$, the problem 
$\partial _t {\mathbf{u}}(t,{\mathbf{x}}) = \mathcal{A}({\mathbf{u}}(t,{\mathbf{x}}), t)$, 
${\mathbf{u}}(0,{\mathbf{x}}) = {\mathbf{u}}_0 ({\mathbf{x}})$, has a unique solution 
${\mathbf{u}}(t,{\mathbf{x}}) \in \mathbb{C}^1 (I;\mathbb{B})$. 
\end{thm} 
\subsection{{The Hilbert Space} ${\bf{SD}}^2$ } 
The purpose of this section is to construct a special class of functions in $\C_c^{\iy}[\R^n]$ (i.e., functions that are infinitely differentiable with compact support).  We will use this class to construct a separable HIlbert space, ${\bf{SD}}^2[\R^n]$, which contains ${\bf{W}}^{k,p}[\R^n]$, for all $k \in \N$ and $1 \le p \le \iy$.
\begin{Def} For $x \in \R, \ 0 \le y < \iy$ and $1<a< \iy$, we define the Jones functions by $g(x, y), \ h(x)$ by (see Jones \cite{JO}, p. 249):
\[
\begin{gathered}
  g(x,y) = \exp \left\{ { - y^a e^{iax} } \right\}, \hfill \\
  h(x) = \left\{ {\begin{array}{*{20}c}
   \begin{gathered}
  \int_0^\infty  {g(x,y)dy} ,\;x \in [ - \tfrac{\pi }
{{2a}},\tfrac{\pi }
{{2a}}], \hfill \\
  0 \quad \quad {\text{            {\rm otherwise} }}. \hfill \\ 
\end{gathered}   \\
   {}  \\

 \end{array} } \right. \hfill \\ 
\end{gathered} 
\]
\end{Def}
The following properties of $g$ are easy to check:
\begin{enumerate}
\item 
\[
\frac{{\partial g(x,y)}}
{{\partial x}} =  - iay^a e^{iax} g(x,y),
\]
\item
\[
\frac{{\partial g(x,y)}}
{{\partial y}} =  - ay^{a - 1} e^{iax} g(x,y),
\]
so that
\item
\[
iy\frac{{\partial g(x,y)}}
{{\partial y}} = \frac{{\partial g(x,y)}}
{{\partial x}}.
\]
\end{enumerate}
It is also easy to see that $h(x)$ is in ${{L}}^1[- \tf{\pi}{2a}, \tf{\pi}{2a}]$ and,
\beqn
\frac{{dh(x)}}
{{dx}} = \int_0^\infty  {\frac{{\partial g(x,y)}}
{{\partial x}}dy}  = \int_0^\infty  {iy\frac{{\partial g(x,y)}}
{{\partial y}}dy}.
\eeqn
Integration by parts in the last expression of equation (8) shows that $h'(x) =  - ih(x)$, so that $
h(x) = h(0)e^{ - ix}$ for $x \in [- \tf{\pi}{2a}, \tf{\pi}{2a}]$.  Since $
h(0) = \int_0^\infty  {\exp \{  - y^a \} dy}$, an additional integration by parts shows that $h(0)= \G(\tf{1}{a} +1)$. For each $l \in \N$ let $a=a_l= 3\times 2^{l-1},\; h(x)=h_l(x), \ x \in [- \tf{\pi}{2a_l}, \tf{\pi}{2a_l}]$ and set $\e_l=\tf{\pi}{4a_l}$. 

  Let $\bQ$ be the set of rational numbers in $\R$ and for each $x^i \in \bQ$,  define
\beqa
f_l^{i} (x) = f_l(x-x^i)=\left\{ {\begin{array}{*{20}c}
   {c_l \exp \left\{ {\frac{{\varepsilon _l^2 }}
{{\left| {x - x^i} \right|^2  - \varepsilon _l^2 }}} \right\},\quad \left| {x - x^i } \right| < \varepsilon _l ,}  \\
   {0,\quad \quad \quad \quad \quad \quad \quad \quad \quad \quad \quad \left| {x - x^i } \right| \geqslant \varepsilon _l ,}  \\

 \end{array} } \right.
\eeqa
where $c_l$ is the standard normalizing constant.  It is easy to check that $f_l^{i} \ne 0$  for $-\e_l<x-x^i < \e_l$, so that the support, ${\rm{spt}}(f_l^{i}) \subset [-\e_l, \e_l]=[-\tf{\pi}{4a_l}, \tf{\pi}{4a_l}]$.

Now set $\chi _l^k (x) =  (f_l^k  * h_l)(x)$, so that ${\rm{spt}}(\chi _l^k ) \subset [-\tf{\pi}{2^{l+1}}, \tf{\pi}{2^{l+1}}]$.  For $x \in {\rm{spt}}(\chi _l^k )$, we can also write $\chi _l^k (x)=\chi _l(x-x^k)$ as:
\[
\int_{ - \infty }^\infty  {f_l[(x-x^k )- z)h_l (z)dz}  = \int_{ - \infty }^\infty  {h_l [(x -x^k)- z]} f_l (z)dz = e^{ - i(x-x^k)} \int_{ - \infty }^\infty  {e^{iz} f_l(z)dz}. 
\]
It is easy to see that $-\pi -\lt|x^k\rt| <x <\pi +\lt|x^k\rt|$.  Thus, if $\alpha_l  = \int_{-\infty} ^\infty  {e^{iz} f_l (z)dz} $ and $I_k={\rm{spt}}(\chi _l^k )$, we can now define:
\[
\xi _l^k (x) = \tfrac{1}
{{n\alpha_l }}\frac{{\chi _l^k [i(x)]}}
{{3^{\pi  + \left| {x^k } \right|} }} = \tfrac{1}
{{n\alpha_l }}\frac{{\chi _l [i(x-x^k)]}}
{{3^{\pi  + \left| {x^k } \right|} }} = \left\{ {\begin{array}{*{20}c}
   {\tfrac{1}
{n}\frac{{e^{(x - x^k )} }}
{{3^{\pi  + \left| {x^k } \right|} }},\quad x \in I_k }  \\
   {{\text{            }}0,\quad x \notin I_k ,}  \\

 \end{array} } \right.
\]
so that $\lt|\xi _l^k (x)\rt| < \tf{1}{n}$.
\subsection{The Space}
To construct our space on $\R^n$, let $\bQ^n$ be the set $
\left\{ {{\mathbf{x}} = (x_1 ,x_2  \cdots ,x_n ) \in {\mathbb{R}}^n } \right\}$ such that $
x_j$ is rational for each $j$.  Since this is a countable dense set in ${\mathbb{R}}^n $, we can arrange it as $\mathbb{Q}^n  = \left\{ {{\mathbf{x}}^1, {\mathbf{x}}^2, {\mathbf{x}}^3,  \cdots } \right\}$.  For each $l$ and $i$, let ${\mathbf{B}}_l ({\mathbf{x}}^i ) $ be the closed cube centered at ${\mathbf{x}}^i$ with edge $\tfrac{\pi }{{a_l }}$ and diagonal of length $r_l=\tfrac{\pi }{{a_l }}\sqrt n$.  

We now choose the  natural order which maps $\mathbb{N} \times \mathbb{N}$ bijectively to $\mathbb{N}$:
\[
\{(1,1), \ (2,1), \ (1,2), \ (1,3), \  (2,2), \  (3,1), \ (3,2), \ (2,3), \  \cdots \}.
\]
Let $\left\{ {{\mathbf{B}}_{k} ,\;k \in \mathbb{N}}\right\}$ be the resulting set of (all) closed cubes 
$
\{ {\mathbf{B}}_l ({\mathbf{x}}^i )\;\left| {(l,i) \in \mathbb{N} \times \mathbb{N}\} } \right.
$
centered at a point in $\mathbb{Q}^n $. For ${\bf{x}} \in {\bf{B}}_k$, let
\beqn
\mcE_k ({\mathbf{x}})  \triangleq \left( {\xi _l^i(x_1) ,\xi _l^i(x_2) , \cdots \xi _l^i(x_n) } \right).
\eeqn
It is easy to see that ${\mathcal{E}}_k ({\mathbf{x}})$ is in ${{{L}}^p [{\mathbb{R}}^n ]}^n \cap {{{L}}^\infty  [{\mathbb{R}}^n ] }^n$ for $1 \le p < \infty$.  Let ${{{L}}^p [{\mathbb{R}}^n ]}^n={{\mathbf{L}}^p [{\mathbb{R}}^n ]}$ and define $F_{k} (\; \cdot \;)$ on $
{{\mathbf{L}}^p [{\mathbb{R}}^n ] }$ by
 \beqn
F_{k} (f) = \int_{{\mathbb{R}}^n } {{\mathcal{E}}_{k} ({\mathbf{x}})\cdot f({\mathbf{x}})d{\mathbf{x}}}. 
\eeqn

It is clear that $F_{k} (\; \cdot \;)$ is a bounded linear functional on ${{\mathbf{L}}^p [{\mathbb{R}}^n ]} $ for each ${k}, \; \left\| {F_{k} } \right\|_\infty   \le 1$.  Furthermore,  if $F_k (f) = 0$ for all ${k}$, $f = 0 \ (a. s.)$, so that $\left\{ {F_{k} } \right\}$ is fundamental on ${{\mathbf{L}}^p [{\mathbb{R}}^n ]} $ for $1 \le p \le \infty$.

Set ${t_{k}}= \tfrac{1}{{2^k }} $ so that ${\sum\nolimits_{k = 1}^\infty  {t_k}}=1$  and  define a  new inner product $\left( {\; \cdot \;} \right) $ on ${{\bf{L}}^2 [{\mathbb{R}}^n ]} $ by
\beqn
 \left( {f,g} \right) = \sum\nolimits_{k = 1}^\infty  {t_k } \left[ {\int_{\mathbb{R}^n } {{\mathcal{E}}_k ({\mathbf{x}})\cdot f({\mathbf{x}})d{\mathbf{x}}} } \right]\left[ {\int_{\mathbb{R}^n } {{\mathcal{E}}_k ({\mathbf{y}}) \cdot g({\mathbf{y}})d{\mathbf{y}}} } \right].    
\eeqn
The completion of ${{\bf{L}}^2 [{\mathbb{R}}^n ]} $, with the above inner product, is also a Hilbert space, ${\bf{SD}}^2 [{\mathbb{R}}^n ] $.   
\begin{rem}
This approach is related to the one used in  \cite{GZ2}, to construct another Hilbert space. Here, one wanted to show that ${{\bf L}^1 [{\mathbb{R}}^n ]}$ can be embedded in a Hilbert space which contains the Denjoy-integrable functions (i.e., functions for which $\left| {\int_{\mathbb{R}^n } {f({\mathbf{x}})d{\mathbf{x}}} } \right| < \infty $) and was used to provide the first rigorous mathematical foundations for the Feynman path integral formulation of quantum mechanics (see also  \cite{GZ2}).  
\end{rem} 
We recall that Alexiewicz \cite{AL} has shown that the class, $D({\R})$, of Denjoy integrable functions (restricted and wide sense) can be normed in the following manner:  for $f \in D({\R})$, define $\left\| f \right\|_D$ by
\beqn
\left\| f \right\|_D  = \sup _s \left| {\int_{ - \infty }^s {f(r)dr} } \right|.
\eeqn
The restricted Denjoy integral is equivalent to the Henstock-Kurzweil integral (see \cite{HS} and \cite{KW}).

Replacing ${\mathbb{R}}$ by ${\mathbb{R}}^n$ in (11), for $f \in D({\mathbb{R}}^n )$, we can also define a norm on $D({\mathbb{R}}^n )$:
\beqn
\left\| f \right\|_D  = \sup _{r > 0} \left| {\int_{{\mathbf{B}}_r } {f({\mathbf{x}})d{\mathbf{x}}} } \right| = \sup _{r > 0} \left| {\int_{\mathbf{R}^n } {{\bf{I}}_{{\mathbf{B}}_r } ({\mathbf{x}})f({\mathbf{x}})d{\mathbf{x}}} } \right| < \infty, 
\eeqn
where ${\mathbf{B}}_r$ is any closed cube of diagonal $r$ centered at the origin in ${\mathbb{R}}^n$, with sides parallel to the coordinate axes and ${\bf{I}}_{{\mathbf{B}}_r}({\mathbf{x}})$ is the indicator function of ${\mathbf{B}}_r $. 
\subsubsection{Functions of Bounded Variation}
The objective of this section is to show that every  HK-integrable function is in  ${\bf{SD}}^2 [{\mathbb{R}}^n ] $.  To do this, we need to discuss a certain class of  functions of bounded variation.  For functions defined on $\R$, the definition of bounded variation is unique.  However, for functions on $\R^n, \; n \ge 2$, there are a number of distinct definitions.
\begin{Def}A function $f \in L^1[\R^n]$ is said to be of bounded variation in the sense of Cesari or $f \in BV_c[\R^n]$, if $f \in L^1[\R^n]$ and each $i, \; 1 \le i \le n$, there exists a signed Radon measure $\mu_i$, such that  
\[
\int_{\mathbb{R}^n } {f({\mathbf{x}})\frac{{\partial \phi ({\mathbf{x}})}}
{{\partial x_i }}d\lambda _n ({\mathbf{x}})}  =  - \int_{\mathbb{R}^n } {\phi ({\mathbf{x}})d\mu _i ({\mathbf{x}})} ,
\]
for all $\phi \in \C_0^\iy[\R^n]$.
\end{Def} 
This is the definition known to most analysts and is the standard one used in geometric measure theory and partial differential equations.  

The class of functions of  bounded variation in the sense of Vitali \cite{YE}, is well known to applied mathematicians and engineers interested in error estimates associated with research in financial derivatives, control theory, robotics, high speed networks and in the calculation of certain integrals.   (See, for example \cite{KAA}, \cite{{NI}}, \cite{PT} or \cite{PTR} and references therein.)

For the general definition, see Yeong (\cite{YE}, p. 175).  We present a definition that is sufficient for continuously  differentiable functions.
\begin{Def} A function $f$ with continuous partials is said to be of bounded variation in the sense of Vitali or $f \in BV_v[\R^n]$ if for all intervals $[a_i,b_i], \, 1 \le i \le n$,
\[
V(f)=\int_{a_1 }^{b_1 } { \cdots \int_{a_n }^{b_n } {\left| {\frac{{\partial ^n f({\mathbf{x}})}}
{{\partial x_1 \partial x_2  \cdots \partial x_n }}} \right|d\lambda _n ({\mathbf{x}})} }  < \infty. 
\]
\end{Def}
\begin{Def}We define $BV_{v,0}[\R^n]$ by: 
\[
BV_{v,0}[\R^n]= \{f({\bf x}) \in BV_v[\R^n]: f({\bf x}) \to 0, \; {\rm as} \; x_i \to -\iy \},
\]
where $x_i$ is any component of ${\bf x}$. 
\end{Def}
The following two theorems may be found in \cite{YE}. (See p. 184 and 187, where the first is used to prove the second.)  If $[a_i, b_i] \subset \R$, we define $[{\bf a}, {\bf b}] \in \R^n$ by   $[{\bf a}, {\bf b}]=\prod_{k=1}^n{[a_i,b_i]}$.  (The notation $(RS)$ means Riemann-Stieltjes.)
\begin{thm}Let $f$ be HK-integrable on $\left[ {{\mathbf{a}},{\mathbf{b}}} \right]$ and let $g \in BV_{v,0}[\R^n]$, then $fg$ is 
HK-integrable and 
\[
(HK)\int_{[{\bf a}, {\bf b}] }{f({\bf x})g({\bf x})d\la_n({\bf x})}=(RS)\int_{[{\bf a}, {\bf b}] }{\lt\{(HK)\int_{[{\bf a}, {\bf x}] }f({\bf y})d\la_n({\bf y})\rt\}dg({\bf x})}
\] 
\end{thm}
\begin{thm}Let $f$ be HK-integrable on $\left[ {{\mathbf{a}},{\mathbf{b}}} \right]$ and let $g \in BV_{v,0}[\R^n]$, then $fg$ is 
HK-integrable and 
\[
\lt|(HK)\int_{[{\bf a}, {\bf b}] }{f({\bf x})g({\bf x})d\la_n({\bf x})}\rt|
\le \lt\|f\rt\|_DV_{[{\bf a}, {\bf b}] }(g)
\]
\end{thm}
\begin{lem}The space $D[{\mathbb{R}}^n ]$, of all HK-integrable functions is contained in ${\bf{SD}}^2 [{\mathbb{R}}^n ] $.
\end{lem}
\begin{proof}
Since each $\mcE_k[{\bf x}]$ is a  continuous and differentiable on its domain, for  $f \in D[{\mathbb{R}}^n ] $, from Theorem 16, we have:
\[
\begin{gathered}
\left\| f \right\|_{{\bf{SD}}^2}^2  = \sum\nolimits_{k = 1}^\infty  {t_k } \left| {\int_{\mathbb{R}^n } {{\mathcal{E}}_k ({\mathbf{x}})\cdot f({\mathbf{x}})d{\mathbf{x}}} } \right|^2  \leqslant \sup _k \left| {\int_{\mathbb{R}^n } {{\mathcal{E}}_k ({\mathbf{x}})\cdot f({\mathbf{x}})d{\mathbf{x}}} } \right|^2  \hfill \\
\leqslant \left\| f \right\|_D^2 [\sup _k V(\mcE_k)]^2 <\iy, \hfill \\
\end{gathered} 
\]
so that $f \in {\bf{SD}}^2 [{\mathbb{R}}^n ]$. 
\end{proof}
\subsubsection{Properties of ${\bf{SD}}^2$}
We now discuss the general properties of ${\mathbf{L}}^p [\R^n ]$.  The first two parts of the following theorem are natural, but the last part is an unexpected benefit. It means that a weakly convergent sequence in any of the  ${\mathbf{L}}^p [\R^n ]$ spaces is strongly convergent in ${\bf{SD}}^2 [{\mathbb{R}}^n ]$.
\begin{thm} For each $p,\;1 \leqslant p \leqslant \infty,\; \mathbf{SD}^2 [{\R}^n ] \supset {\mathbf{L}}^p [\R^n ]$ as a dense, continuous and compact  embedding.
\end{thm}
\begin{proof} First, by construction, $\mathbf{SD}^2 [\R^n ]$ contains ${\mathbf{L}}^2 [\R^n ]$ densely, so we need only show that $\mathbf{SD}^2 [\R^n ] \supset {\mathbf{L}}^q [\R^n ]$ for $q \ne 2$.  If $f \in {\mathbf{L}}^q [\R^n ]$ and $q < \infty $, we have
\[
\begin{gathered}
  \left\| f \right\|_{{\mathbf{SD}}^2 }  = \left\{ {\sum\nolimits_{k = 1}^\infty  {t_k \left| {\int_{\mathbb{R}^n } {\mcE_k ({\mathbf{x}}) \cdot f({\mathbf{x}})d{\mathbf{x}}} } \right|^2 } } \right\}^{1/2}    \hfill \\
\leqslant \left\{ {\sum\nolimits_{k = 1}^\infty  {t_k \left[ {\int_{\mathbb{R}^n } {\left| {\mcE_k ({\mathbf{x}})} \right|^q  \cdot \left| {f({\mathbf{x}})} \right|^q d{\mathbf{x}}} } \right]^{\tfrac{2}
{q}} } } \right\}^{1/2}  \hfill \\
   \leqslant \mathop {\sup }\limits_k \left\{ {\left[ {\int_{\mathbb{R}^n } {\left| {\mcE_k ({\mathbf{x}})} \right|^q  \cdot \left| {f({\mathbf{x}})} \right|^q d{\mathbf{x}}} } \right]^{\tfrac{1}
{q}} } \right\} \leqslant \mathop {\sup }\limits_k \left\| {\mcE_k } \right\|_q \left\| f \right\|_q  \leqslant \left\| f \right\|_q . \hfill \\ 
\end{gathered} 
\] 
In the last term, we used  ${\sup_k}\lt\|{\mathcal{E}}_k \rt\|_q <1$, so that $f \in \mathbf{SD}^2 [\R^n ] $.   if $q = \infty $, we have 
\[
\begin{gathered}
  \left\| f \right\|_{{\mathbf{SD}}^2 }  = \left\{ {\sum\nolimits_{k = 1}^\infty  {t_k \left| {\int_{\mathbb{R}^n } {\mcE_k ({\mathbf{x}}) \cdot f({\mathbf{x}})d{\mathbf{x}}} } \right|^2 } } \right\}^{1/2}  \hfill \\
   \leqslant \mathop {\sup }\limits_k \left\{ {\left( {\int_{\mathbb{R}^n } {\left| {\mcE_k ({\mathbf{x}})} \right|\left| {f({\mathbf{x}})} \right|d{\mathbf{x}}} } \right)} \right\} \leqslant \mathop {\sup }\limits_k \left\| {\mcE_k } \right\|_1 \left\| f \right\|_\iy  \leqslant \left\| f \right\|_\iy , \hfill \\ 
\end{gathered} 
\]
since  $\left\| {\mcE_k } \right\|_{L^1} <1$.

The proof of compactness follows from the fact that, if $\{f_n \}$ is any weakly convergent sequence in  ${\bf{L}}^{p} [{\mathbb{R}}^n ], \ 1 \le p < \iy$ with limit $f$, then since ${\mathcal{E}}_k ({\mathbf{x}}) \in {\bf{L}}^{q} [{\mathbb{R}}^n ], \ 1 \le q \le \iy$,
\[
\int_{\mathbb{R}^n } { {\mathcal{E}}_k ({\mathbf{x}})\cdot \left[ {f_n ({\mathbf{x}}) - f({\mathbf{x}})} \right]d{\mathbf{x}}}  \to 0
\]
for each $k$.  Thus,  $\{f_n \}$ converges strongly to $f$ in ${\bf{SD}}^2[{\mathbb{R}}^n ]$.

Finally, we note that ${\bf{SD}}^2 [{\R}^n ] \supset {\bf{L}}^1 [{\R}^n ]^{ *  * } {\kern 1pt}  = \mathfrak{M}[{\R}^n ]$, the space of finitely additive measures on $\R^n$.  It follows that $d\mu_k({\mathbf{x}}) =\mcE_k ({\mathbf{x}})d{\bf x}$ defines an element in $\mathfrak{M}[{\R}^n ]$ (the dual space of ${\bf{L}}^{\iy} [{\mathbb{R}}^n ]$). Thus, if  $\{f_n \}$ is any weakly convergent sequence to $f$ in ${\bf{L}}^{\iy} [{\mathbb{R}}^n ], \; \{f_n \}$ converges strongly to $f \in {\bf{SD}}^2[{\mathbb{R}}^n ]$.    
\end{proof}
\begin{rem}
Since ${\bf{L}}^\infty  [{\mathbb{R}}^n ] \subset {\bf{SD}}^2[{\mathbb{R}}^n ]$, while $
{\bf{SD}}^2 [{\mathbb{R}}^n ] $ is separable, we  see in a clear  and forceful manner that separability is not an inherited property.
\end{rem}
\begin{Def}We call ${\bf{SD}}^2 [{\mathbb{R}}^n ] $ the strong distribution Hilbert space for $\R^n$.
\end{Def}
In order to justify our definition, let $\al$ be a multi-index of nonnegative integers, $\al = (\al_1, \ \al_2, \ \cdots \ \al_n)$, with $\left| \alpha  \right| = \sum\nolimits_{j = 1}^n {\alpha _j } $.   If $D$ denotes the standard partial differential operator, let 
\[
D^{\al}=D^{\al_1}D^{\al_2} \cdots D^{\al_k}.
\]
\begin{thm} If ${\bf u} \in {\bf{SD}}^2[{\mathbb{R}}^n ]$ and $D^{\al}{\bf u}={\bf v}_{\al}$ in the weak  sense, then ${\bf v}_{\al} \in {\bf{SD}}^2[{\mathbb{R}}^n ]$.
\end{thm}
\begin{proof} From our construction, each ${\mathcal{E}}_k \in \C_c^{\iy}[{\mathbb{R}^n }]$, so that
\[ 
{\int_{\mathbb{R}^n } {{\mathcal{E}}_k ({\mathbf{x}}) \cdot D^{\al}{{\bf u}}({\mathbf{x}})d{\mathbf{x}}} }= (-1)^{\lt|\al \rt|}{\int_{\mathbb{R}^n } D^{\al} {{\mathcal{E}}_k ({\mathbf{x}}) \cdot{{\bf v}_{\al}}({\mathbf{x}})d{\mathbf{x}}} }.
\]
An easy calculation shows that, for any $j, \; {\int_{\mathbb{R}^n } {\partial_j{\mathcal{E}}_k ({\mathbf{x}}) \cdot {{\bf u}}({\mathbf{x}})d{\mathbf{x}}} }= {\int_{\mathbb{R}^n } {{\mathcal{E}}_k ({\mathbf{x}}) \cdot{{\bf u}_{\al}}({\mathbf{x}})d{\mathbf{x}}} }$, so that 
\[ 
{\int_{\mathbb{R}^n } {{\mathcal{E}}_k ({\mathbf{x}}) \cdot D^{\al}{{\bf u}}({\mathbf{x}})d{\mathbf{x}}} }= (-1)^{\lt|\al \rt|}{\int_{\mathbb{R}^n } {{\mathcal{E}}_k ({\mathbf{x}}) \cdot{{\bf v}_{\al}}({\mathbf{x}})d{\mathbf{x}}} }.
\]
It now follows that, for any $g  \in {\bf{SD}}^2[{\mathbb{R}}^n ], \; (D^{\al}{{\bf u}}, g)_{{\bf{SD}}^2}= (-1)^{\lt|\al \rt|}({{\bf v}_{\al}},g)_{{\bf{SD}}^2}$, so that  ${\bf v}_{\al} \in {\bf{SD}}^2[{\mathbb{R}}^n ]$.
\end{proof} 
The next result follows from Theorem 21 and explains our use of the term strong distribution in describing ${\bf{SD}}^2[{\mathbb{R}}^n ]$.  
\begin{cor} If ${\bf u}$ is in the domain of $D^{\al}$, then for any $g  \in {\bf{SD}}^2[{\mathbb{R}}^n ], \; (D^{\al}{{\bf u}}, g)_{{\bf{SD}}^2}= (-1)^{\lt|\al \rt|}({\bf u},g)_{{\bf{SD}}^2}$ so that, in particular, $\lt\|D^{\al}{{\bf u}}\rt\|_{{\bf{SD}}^2} = \lt\|{{\bf u}}\rt\|_{{\bf{SD}}^2}$.
\end{cor}
Recall that a function ${\bf u}$ is said to be in ${\bf W}^{k,p}[{\mathbb{R}}^n ], \; k \in \N, \; 1 \le p \le \iy$, if
\[
\begin{gathered}
  \left\| {\bf u} \right\|_{k,p}  = \left\{ {\sum\limits_{0 \leqslant \left| \alpha  \right| \leqslant k} {\left\| {D^\alpha  {\bf u}} \right\|_{{\bf L}^p }^p } } \right\}^{1/p}  < \infty ,\;\;{\text{if}}\quad 1 \leqslant p < \infty , \hfill \\
  \left\| {\bf u} \right\|_{k,\infty }  = \mathop {\max }\limits_{0 \leqslant \left| \alpha  \right| \leqslant k} \left\| {D^\alpha  {\bf u}} \right\|_{{\bf L}^\infty  }  < \infty ,\quad {\text{if}}\quad p = \infty . \hfill \\ 
\end{gathered}
\]
\begin{lem} For any $p, \; 1 \le p \le \iy$ and all $k \in \N, \; {\bf W}^{k, \, p}[{\mathbb{R}}^n ]  \subset {\bf{SD}}^2[{\mathbb{R}}^n ]$.
\end{lem}
\begin{proof} If $u \in {\bf W}^{k, \, p}[{\mathbb{R}}^n ]$, then, for any $\al, \ 0 \leqslant \left| \alpha  \right| \leqslant k$, we have
\[
\left\| u \right\|_{{\mathbf{SD}}^2 }^p  = \left\| {D^\alpha  u} \right\|_{{\mathbf{SD}}^2 }^p  \le \sum\limits_{0 \leqslant \left| \alpha  \right| \leqslant k} {\left\| {D^\alpha  u} \right\|_{{\mathbf{SD}}^2 }^p } \le  \sum\limits_{0 \leqslant \left| \alpha  \right| \leqslant k} {\left\| {D^\alpha  u} \right\|_{{\mathbf{L}}^p }^p }. 
\]
It follows that $u \in {\bf{SD}}^2[{\mathbb{R}}^n ]$.  The case of $p=\iy$ is clear.
\end{proof}
\subsection{The Nonlinear Term: A Priori Estimates}
The difficulty in proving the existence of global-in-time strong solutions for equation (4) is directly linked to the problem of getting good a priori estimates for the nonlinear term ${{B}}({\mathbf{u}},{\mathbf{u}})$.  
\begin{thm}   If $\bf A$ is the Stokes operator and  ${\bf u}({\bf x},t) \in D({\bf{A}})$ is a reasonable vector field, then
\begin{enumerate}
\item ${\left\langle { -\nu{\bf{A}}{\mathbf{u}},{\mathbf{u}}} \right\rangle _{{\mathbb{H}}_{sd} }} = -\nu \left\| { {\bf A}{\mathbf{u}} } \right\|_{ {\mathbb{H}}_{sd} }^2$.
\item For ${\bf u}({\bf x},t) \in {\bf SD}^2 \cap D({\bf{A}})$ and each $t \in [0, \iy)$, there exists a constant $M =M({\bf u}({\bf x},0))>0$, such that
\beqn
\left| {\left\langle {B({\mathbf{u}},{\mathbf{u}}),{\mathbf{u}}} \right\rangle _{{\mathbb{H}}_{sd} }} \right| \le M \left\| { {\mathbf{u}} } \right\|_{ {\mathbb{H}}_{sd} }^3. 
\eeqn
\item
\beqn
\left| {\left\langle {{{B}}({\mathbf{u}},{\mathbf{v}}),{\mathbf{w}}} \right\rangle _{{\mathbb{H}}_{sd}} } \right| \le M \left\| {\bf{u}} \right\|_{{\mathbb{H}}_{sd}} \left\| {\bf{w}} \right\|_{{\mathbb{H}}_{sd}} \left \| {\bf{v}} \right\|_{{\mathbb{H}}_{sd}}. 
\eeqn
\item
\beqn
max\{ \left\| {{{B}}({\mathbf{u}},{\mathbf{v}})} \right\|_{{\mathbb{H}}_{sd}}, \ \left\| {{{B}}({\mathbf{v}},{\mathbf{u}})} \right\|_{{\mathbb{H}}_{sd}} \} \leqslant M \left\| {\mathbf{u}} \right\|_{{\mathbb{H}}_{sd}} \left\| {\mathbf{v}} \right\|_{{\mathbb{H}}_{sd}}. 
\eeqn
\end{enumerate}
\end{thm}
\begin{proof} From equation (12), we have
\beqa
{\left\langle {-\nu {\bf{A}}{\mathbf{u}},{\mathbf{u}}} \right\rangle _{{\mathbb{H}}_{sd} }}=-\nu
 \sum\nolimits_{k = 1}^\infty  {t_k } \left[ {\int_{\mathbb{R}^n } {{\mathcal{E}}_k ({\mathbf{x}})\cdot {\bf{A}}{\mathbf{u}}({\mathbf{x}})d{\mathbf{x}}} } \right]\left[ {\int_{\mathbb{R}^n } {{\mathcal{E}}_k ({\mathbf{y}})\cdot {\mathbf{u}}({\mathbf{y}})d{\mathbf{y}}} } \right].  
\eeqa
Using the fact that ${\bf{u}} \in D({\bf A})$ and that $k=(l,i)$ (see equation (9)), so that
\[
\begin{gathered} 
{\int_{\mathbb{R}^n } {{\mathcal{E}}_k ({\mathbf{y}}) \cdot {\partial _{y_j }^2 }{\mathbf{u}}({\mathbf{y}})d{\mathbf{y}}} }= {\int_{\mathbb{R}^n } {\partial _{y_j }^2} {{\mathcal{E}}_k ({\mathbf{y}}) \cdot{\mathbf{u}}({\mathbf{y}})d{\mathbf{y}}} } \hfill \\
= {\int_{I_{i} } {\partial _{y_j }^2} {\left( {\xi _l^i(y_1) ,\xi _l^i(y_2) , \cdots \xi _l^i(y_n) } \right) \cdot{\mathbf{u}}({\mathbf{y}})d{\mathbf{y}}} }={\int_{\mathbb{R}^n } {{\mathcal{E}}_k ({\mathbf{y}}) \cdot {\mathbf{u}}({\mathbf{y}})d{\mathbf{y}}} }.  \hfill \\
\end{gathered}
\]
Using this in the above equation and summing on $j$, we have 
\[ 
{\int_{\mathbb{R}^n } {{\mathcal{E}}_k ({\mathbf{y}}) \cdot {\bf A }{\mathbf{u}}({\mathbf{y}})d{\mathbf{y}}} } ={\int_{\mathbb{R}^n } {{\mathcal{E}}_k ({\mathbf{y}}) \cdot{\mathbf{u}}({\mathbf{y}})d{\mathbf{y}}} }.
\]
It follows that
\beqa
{\left\langle { {\bf{A}}{\mathbf{u}},{\mathbf{u}}} \right\rangle _{{\mathbb{H}}_{sd} }}=
 \sum\nolimits_{k = 1}^\infty  {t_k } \left[ {\int_{\mathbb{R}^n } {{\mathcal{E}}_k ({\mathbf{x}})\cdot {\bf{A}}{\mathbf{u}}({\mathbf{x}})d{\mathbf{x}}} } \right]\left[ {\int_{\mathbb{R}^n } {{\mathcal{E}}_k ({\mathbf{y}})\cdot {\bf{A}}{\mathbf{u}}({\mathbf{y}})d{\mathbf{y}}} } \right] = \left\| { {\bf A}{\mathbf{u}} } \right\|_{ {\mathbb{H}}_{sd} }^2.  
\eeqa
This proves (1).  To prove (2),  let   $\vec \delta ({\mathbf{x}}) = \left( {\delta (x_1 ), \cdots \delta_k (x_3 )} \right)$, the $n$-dimensional Dirac delta function and set $\hat{\e}=\lt\|\vec \delta ({\mathbf{x}})\rt\|_{\mathbb{H}_{sd} }$.
We start with  
\[
b({\bf u}, {\bf u}, {\mathcal{E}}_k)= \left| {\left\langle {B({\mathbf{u}},{\mathbf{u}}),{\mathcal{E}}_k} \right\rangle _{{\mathbb{H}}_{sd} }} \right|=
\left| {\int_{\mathbb{R}^3 } {\left( {{\mathbf{u}}({\mathbf{x}}) \cdot \nabla } \right){\mathbf{u}}({\mathbf{x}}) \cdot {\mathcal{E}}_k({\mathbf{x}})d{\mathbf{x}}} } \right|
\]
and integrate by parts, to get 
\[
\left| {\int_{\mathbb{R}^3 } {\left\{ {\sum\nolimits_{i = 1}^3 {u_i ({\mathbf{x}})^2 \mcE_k^i ({\mathbf{x}})d{\mathbf{x}}} } \right\}} } \right| \leqslant \mathop {\sup }\limits_k \left\| {\mcE_k } \right\|_\infty  \left\| {\mathbf{u}} \right\|_\Ha^2  \le  \left\| {\mathbf{u}} \right\|_\Ha^2 .
\]
Since ${\bf{u}}$ is reasonable, there is a constant  $\bar{M}$  depending on ${\bf{u}}(0)$ and $f$,  such that  $\left\| {\mathbf{u}} \right\|_2^2 \le \bar{M} \left\| {\mathbf{u}} \right\|_{\mathbb{H}_{sd} }^2$.   We now have
\[
\begin{gathered}
  \left| {\left\langle {B({\mathbf{u}},{\mathbf{u}}),{\mathbf{u}}} \right\rangle _{\mathbb{H}_{sd} } } \right| = \left| {\sum\nolimits_{k= 1}^\infty  {t_k } \left[ {\int_{\mathbb{R}^3 } {\left( {{\mathbf{u}}({\mathbf{x}}) \cdot \nabla } \right){\mathbf{u}}({\mathbf{x}}) \cdot \mcE_k ({\mathbf{x}})d{\mathbf{x}}} } \right]\left[ {\int_{\mathbb{R}^3 } {{\mathbf{u}}({\mathbf{y}}) \cdot \mcE_k ({\mathbf{y}})d{\mathbf{y}}} } \right]} \right| \hfill \\
   \leqslant \bar{M} \hat{\e}^{-2} \left\| {\mathbf{u}} \right\|_{\mathbb{H}_{sd} }^2 \left| {\sum\nolimits_{k= 1}^\infty  {t_k } \left[ {\int_{\mathbb{R}^3 } {{\vec \delta({\mathbf{x}})} \cdot \mcE_k ({\mathbf{x}})d{\mathbf{x}}} } \right]\left[ {\int_{\mathbb{R}^3 } {{\mathbf{u}}({\mathbf{y}}) \cdot \mcE_k ({\mathbf{y}})d{\mathbf{y}}} } \right]} \right| \hfill \\
   \leqslant {M}\left\| {\mathbf{u}} \right\|_{\mathbb{H}_{sd} }^3, \hfill \\ 
\end{gathered} 
\]
where $M=\bar{M} \hat{\e}^{-1}$ and the third line above follows from Schwartz's inequality.  The proofs of (3) and (4) are easy.
\end{proof}
\subsection{Generation Theorem}
  We now begin with a study of the operator $ {\mathcal{A}}( \cdot ,t)$, for fixed $t$, and establish conditions depending on ${\mathbf{A}},  {\text{ }}\nu ,{\text{ }}  {\text{ and }}{\mathbf{f}}(t)$ which guarantee that $ {\mathcal{A}}( \cdot ,t)$ generates a contraction semigroup.  Clearly ${\mathcal{A}}( \cdot ,t)$ is defined on $D({\bf{A}}) $ and, since $ \nu \mathbf{A}$  is a closed positive (m-accretive) operator, $ - \nu {\mathbf{A}}$ generates a linear contraction semigroup. Thus, we need to ensure that $ {\mathcal{A}}( \cdot ,t)$ will be m-dissipative for each $t$.  We assume that 
$
{\mathbf{f}}(t) \in L^\infty[[0,\infty); {\mathbb H}_{sd}]
$
and is H\"{o}lder continuous in $t$, with $\left\| {{\mathbf{f}}(t) - {\mathbf{f}}(\tau )} \right\|_{{\mathbb{H}}_{sd}}  \le a\left| {t - \tau } \right|^\theta,{\text{ }}a > 0,{\text{ }}0 < \theta  < 1$.

\begin{thm} If $0 \ne f = \sup _{t \in {\mathbf{R}}^ +  } \left\| {\mathbb{P}{\mathbf{f}}(t)} \right\|_{{\mathbb{H}}_{sd}}  < \infty $,  there exist positive constants ${{u}}_ +, \; {u}_-  $, depending only on $f$, ${\mathbf{A}}$ and $\nu $  such that, for all ${\mathbf{u}}$ with $
0 < {u}_- \le \left\| {\mathbf{u}} \right\|_{{\mathbb{H}}_{sd}}  \le \tf{1}{2}{{u}}_ +, \;  {\mathcal{A}}( \cdot ,t)$ is strongly dissipative. 

If $f=0$ this implies that $u_{-}=0$.  In this case, we replace $\tf{1}{2}u_+$ by $\tf{(1-\e)}{2}u_+$, so that $ {\mathcal{A}}( \cdot ,t)$ is strongly dissipative on $
0 <  \left\| {\mathbf{u}} \right\|_{{\mathbb{H}}_{sd}}  \le \tf{(1-\e)}{2}{{u}}_ + $.
\end{thm}
\begin{proof} The proof of our assertion has two parts. First, for $f \ne 0$, we require that the nonlinear operator $ {\mathcal{A}}( \cdot ,t)$
 be 0-dissipative, which gives us an upper bound ${{u}}_ +  $ and lower bound ${{u}}_ -  $ in terms of the norm (i.e., $\left\| {\mathbf{u}} \right\|_{{\mathbb{H}}_{sd}}  \leqslant {{u}}_ + $ ).  We then use this part to show that $ {\mathcal{A}}( \cdot ,t)$ is strongly dissipative on $D({\bf{A}}) \cap \B$, for any closed convex  set, $ \mathbb{B}$, inside the annulus defined by $ \left\{ {{\mathbf{u}} \in D({\bf{A}}):0 \le {{u}}_ -  \le \left\| {\mathbf{u}} \right\|_{{\mathbb{H}}_{sd}}  \leqslant \tfrac{1}{2} {{u}}_ +  } \right\}$.  We then consider adjustments when $f=0$.

Part 1) 
From equation (4), we consider the expression
\beqa
\begin{gathered}
  \left\langle {{\mathcal A}({\mathbf{u}},t),{\mathbf{u}}} \right\rangle _{{\mathbb{H}}_{sd}}  =  - \nu \left\langle {{\mathbf{Au}},{\mathbf{u}}} \right\rangle _{{\mathbb{H}}_{sd}}  + \left\langle {\left[ { - B({\mathbf{u}},{\mathbf{u}}) + \mathbb{P}{\mathbf{f}}} \right],{\mathbf{u}}} \right\rangle _{{\mathbb{H}}_{sd}}  \hfill \\
   =  - \nu \left\| {{\mathbf{A}} {\mathbf{u}}} \right\|_{{\mathbb{H}}_{sd}}^2  - \left\langle {B({\mathbf{u}},{\mathbf{u}}),{\mathbf{u}}} \right\rangle _{\Ha_{sd}}  + \left\langle {\mathbb{P}{\mathbf{f}},{\mathbf{u}}} \right\rangle _{\Ha_{sd}}.  \hfill \\ 
\end{gathered} 
\eeqa
It follows that
\[
  \left\langle {{\mcA}({\mathbf{u}},t),{\mathbf{u}}} \right\rangle_{{\mathbb{H}}_{sd}}  \le   - \nu  \left\| {\mathbf{u}} \right\|_{\Ha_{sd}}^2  +  \bar{M}\left\| {{\mathbf{u}}} \right\|_{\Ha_{sd}}^3  +  f \left\| { {\mathbf{u}}} \right\|_{\Ha_{sd}}.  
\] 
Since $\left\| {\mathbf{u}} \right\|_{\Ha_{sd}}  > 0$, we have that ${\mcA}( \cdot ,t)$ is 0-dissipative if
\beqn
 - \nu \left\| {\mathbf{u}} \right\|_{\Ha_{sd}}  + \bar{M} \left\| {\mathbf{u}} \right\|_{\Ha_{sd}}^2  + f \leqslant 0
\eeqn
If ${\mathbf{v}} \in D({\mathbf{A}})$ is a reasonable vector field satisfying (17), we let $M = \sup \left\{ {\bar{M}_{\mathbf{v}} }\right\}$.  If we solve  inequality (17) using $M$, we get  
\beqa
{{u}}_ \pm   = \tfrac{  \nu }
{ 2M} \left\{ {1 \pm \sqrt {1 - ({{4f M )} \mathord{\left/
 {\vphantom {{4fM \de)} {(\nu )}^2}} \right.
 \kern-\nulldelimiterspace} {(\nu )}^2}}} \right\} = \tfrac{ \nu}
{ 2M} \left\{ {1 \pm \sqrt {1 - \gamma } } \right\},
\eeqa
where $
\gamma  = {{(4 fM)} \mathord{\left/
 {\vphantom {{(4 fM)} { \nu ^2}}} \right.
 \kern-\nulldelimiterspace} { \nu ^2 }}.
$
Since we want real distinct solutions, we must require that 
\beqn
\gamma  = \frac{{4 fM}}
{{ \nu ^2 }} < 1\; \Rightarrow 2{{\sqrt {fM } }} < \nu. 
\eeqn
It is clear that, if $\mathbb{P}{\mathbf{f}} \ne {\mathbf{0}}$, then 
$
{{u}}_ -   < {{u}}_ + $ , and our requirement that $\mcA({\mathbf{u}},t)$ is 0-dissipative implies that, since our solution factors as 
$
(\left\| {\mathbf{u}} \right\|_{\Ha_{sd}}  - {{u}}_ +  )(\left\| {\mathbf{u}} \right\|_{\Ha_{sd}}  - {{u}}_ -  ) \le 0,
$
we must have that:
\beqa
\left\| {\mathbf{u}} \right\|_{\Ha_{sd}}  - {{u}}_ +   \le 0,{\text{  }}\left\| {\mathbf{u}} \right\|_{\Ha_{sd}}  - {{u}}_ -   \ge 0.
\eeqa
 It follows that, for  
${{u}}_ -   \le \left\| {\mathbf{u}} \right\|_{\Ha_{sd}}  \le {{u}}_ + $, 
$
\left\langle {\mcA({\mathbf{u}},t),{\mathbf{u}}} \right\rangle _{\Ha_{sd}}  \le 0$.  (It is clear that, when $
\mathbb{P}{\mathbf{f}}(t) = {\mathbf{0}}, \; {{u}}_ -   = {{0}}$, and ${{u}}_ +   = \tfrac{\nu}{M} $.)

Part 2): Now, for any ${\mathbf{u}},{\mathbf{v}} \in D({\mathbf{A}})$  with ${\mathbf{u}}-{\mathbf{v}} \in D({\mathbf{A}})$ and 
\[
u_- \le \min ({\text{ }}\left\| {\mathbf{u}} \right\|_{\Ha_{sd}} ,\left\| {\mathbf{v}} \right\|_{\Ha_{sd}} ) \le \max ({\text{ }}\left\| {\mathbf{u}} \right\|_{\Ha_{sd}} ,\left\| {\mathbf{v}} \right\|_{\Ha_{sd}} ) \le (1/2){{{u}}_ +},
\]
 we have that   
\beqn
\begin{gathered}
  \left\langle {{\mcA}({\mathbf{u}},t) - {\mcA}({\mathbf{v}},t),({\mathbf{u}} - {\mathbf{v}})} \right\rangle _{\Ha_{sd}}  =  -\nu \left\| {{\bf{A}} ({\mathbf{u}} - {\mathbf{v}})} \right\|_{\Ha_{sd}}^2  \hfill \\
  {\text{                                                    }} -  \left\langle { [{{B}}({\mathbf{u}},{\mathbf{u}} - {\mathbf{v}}) + {{B}}({\mathbf{v}}, {\mathbf{u-v}})],({\mathbf{u}} - {\mathbf{v}})} \right\rangle _{\Ha_{sd}}  \hfill \\  
\end{gathered}  
\eeqn
\beqa
\begin{gathered}
  {\text{                    }} \leqslant - \nu  \left\| {{\mathbf{u}} - {\mathbf{v}}} \right\|_{\Ha_{sd}}^2 +  M \left\| {{\mathbf{u}} - {\mathbf{v}}} \right\|_{{\mathbb{H}}_{sd}}^2 \left( {\left\| {\mathbf{u}} \right\|_{{\mathbb{H}}_{sd}}  + \left\| {\mathbf{v}} \right\|_{\Ha_{sd}} } \right) \hfill \\
  {\text{                    }} \le - \nu  \left\| {{\mathbf{u}} - {\mathbf{v}}} \right\|_{{\mathbb{H}}_{sd}}^2  + M  \left\| {{\mathbf{u}} - {\mathbf{v}}} \right\|_{\Ha_{sd}}^2 {{u}}_ +   \hfill \\
  {\text{                    }} =  -  \nu  \left\| {{\mathbf{u}} - {\mathbf{v}}} \right\|_{{\mathbb{H}}_{sd}}^2  +  M  \left\| {{\mathbf{u}} - {\mathbf{v}}} \right\|_{\Ha_{sd}}^2 \left( \tfrac{\nu}{2M} \left\{ {1 + \sqrt {1 - \gamma } } \right\} \right) \hfill \\
  {\text{                    }} =  - \tf{\nu }{2} \left\| {{\mathbf{u}} - {\mathbf{v}}} \right\|_{\Ha_{sd}}^2 \left\{ {1 - \sqrt {1 - \gamma } } \right\} \hfill \\
  {\text{                    }} =  - \s \left\| {{\mathbf{u}} - {\mathbf{v}}} \right\|_{\Ha_{sd}}^2 ,{\text{   }} \s = \tf{\nu}{2}  \left\{ {1 - \sqrt {1 - \gamma } } \right\}. \hfill \\ 
\end{gathered} 
\eeqa
If $f=0$, in the computation above,  we see that $\s=0$.  To obtain our result in this case,  we replace $\tf{1}{2}u_{+}$ by $\tf{(1-\e)}{2}u_{+}$.  The same computation shows that $\s=\nu \e$
\end{proof} 
Let $\B$ be any closed convex set inside the annulus bounded by $\tfrac{1}{2}{{u}}_+$ and ${{u}}_{-}$. The first part of  the next Lemma follows easily from the properties of ${\bf f}(t)$,  the second part follows from Part 2) of Theorem 24 (see equation 18), while the third part is trivial. 
\begin{lem} Let $t, \; \tau \in I=[0, \infty)$ and ${\mathbf{u}},  \; {\bf v} \in D({\bf{A}})$.  Then
\begin{enumerate}
 \item The mapping ${\mathcal{A}}({\mathbf{u}},t)$ is H\"{o}lder continuous in $t$, with $
\left\| { {\mathcal{A}}({\mathbf{u}},t) -  {\mathcal{A}}({\mathbf{u}},\tau )} \right\|_{\Ha_{sd}}  \le a\left| {t - \tau } \right|^\theta$, where  $a$ is the H\"{o}lder constant for the function ${\mathbf{f}}(t)$. 
 \item The mapping ${\mathcal{A}}({\mathbf{u}},t)$ is a Lipschitz continuous function in ${\bf u}$, with 
\[
\left\| { {\mathcal{A}}({\mathbf{u}},t) -  {\mathcal{A}}({\mathbf{v}},t )} \right\|_{\Ha_{sd}}  \le b \left\| {\bf u} -{\bf v}\right\|_{\Ha_{sd}}.
\]
\item The mapping ${\mathcal{A}}({\mathbf{u}},t) -\Pa{\bf f}(t)$ is coercive:
\[
\mathop {\lim }\limits_{\left\| {\mathbf{u}} \right\|_{\Ha_{sd}} \to \iy}  \frac{{\left\langle {\left( {\mcA({\mathbf{u}},t) - \mathbb{P}{\mathbf{f}}(t)} \right),{\mathbf{u}}} \right\rangle _{\Ha_{sd}} }}
{{\left\| {\mathbf{u}} \right\|_{\Ha_{sd}} }} = \infty .
\]
\end{enumerate}
\end{lem}

\begin{thm} The operator ${\mathcal{A}}(\cdot,t) $
 is strongly dissipative and jointly continuous in ${\mathbf{u}}$ and $t$.  Furthermore, for each $t \in {\mathbf{R}}^ +  $ and $\beta  > 0$, $Ran[I - \beta  {\mathcal{A}}(t)] \supset \B$, so that ${\mathcal{A}}(t)$ is m-dissipative on $\B$. 
\end{thm}
\begin{proof} 
From Theorem 25, ${\mcA}( \cdot ,t)$ is strongly dissipative.  A strongly dissipative operator is maximal dissipative, so that $Ran[I - \beta  {\mathcal{A}}(\cdot, t)] \supset \B$.  It follows from \cite{CP} that, since ${\mathbb{H}}_{\Ha_{sd}}$ is a Hilbert space,  ${\mcA}( \cdot ,t)$ is m-dissipative on $\B$ for each $t \in {\mathbf{R}}^ + $.  
 
To see that ${\mcA}( {\bf u} ,t)$ is continuous in both variables, let ${\mathbf{u}}_n ,{\mathbf{u}} \in \D$, $\left\| ({\mathbf{u}}_n  - {\mathbf{u}}) \right\|_{\Ha_{sd}}   \to 0$,
 with $t_n ,t \in I$ and $t_n  \to t$.    Using  $\left\| {{\mathbf{Au}}} \right\|_{\Ha_{sd}}  = \left\| {\mathbf{u}} \right\|_{\Ha_{sd}}$, we have
\beqa
\begin{gathered}
  \left\| { {\mathcal{A}}({\mathbf{u}}_n, t_n )  -  {\mathcal{A}}({\mathbf{u}}, t)} \right\|_{\Ha_{sd}}  \leqslant \left\| { {\mathcal{A}}({\mathbf{u}}, t_n ) -  {\mathcal{A}}({\mathbf{u}}, t) }\right\|_{\Ha_{sd}}   + \left\| { {\mathcal{A}}({\mathbf{u}}_n, t_n )  -  {\mathcal{A}}({\mathbf{u}}, t_n )} \right\|_{\Ha_{sd}}   \hfill \\
   = \left\| {{\text{[}}\mathbb{P}{\mathbf{f}}(t_n ) - \mathbb{P}{\mathbf{f}}(t)]} \right\|_{\Ha_{sd}}   + \left\| {\nu {\mathbf{A}}({\mathbf{u}}_n  - {\mathbf{u}}) + [{{B}}({\mathbf{u}}_n  - {\mathbf{u}},{\mathbf{u}}_n ) + {{B}}( {\mathbf{u}},{\mathbf{u}}_n  -{\mathbf{u}})]} \right\|_{\Ha_{sd}}   \hfill \\
   \leqslant d\left| {t_n  - t} \right|^\theta   + \nu \left\| {{\mathbf{A}}({\mathbf{u}}_n - {\mathbf{u}})} \right\|_{\Ha_{sd}}   + \left\| {{{B}}({\mathbf{u}}_n  - {\mathbf{u}},{\mathbf{u}}_n ) + {{B}}( {\mathbf{u}},{\mathbf{u}}_n  - {\mathbf{u}})} \right\|_{\Ha_{sd}}   \hfill \\
   \leqslant d\left| {t_n  - t} \right|^\theta   + \nu  \left\| {({\mathbf{u}}_n  - {\mathbf{u}})} \right\|_{\Ha_{sd}}  +  {{M}} \left\| {({\mathbf{u}}_n  - {\mathbf{u}})} \right\|_{\Ha_{sd}} \left\{ {\left\| {{\mathbf{u}}_n } \right\|_{\Ha_{sd}}   + \left\| {{\mathbf{u}}} \right\|_{\Ha_{sd}}  } \right\} \hfill \\
   \leqslant d\left| {t_n  - t} \right|^\theta   + \nu  \left\| {({\mathbf{u}}_n  - {\mathbf{u}})} \right\|_{\Ha_{sd}}   +  
{{M}} \left\| {({\mathbf{u}}_n  - {\mathbf{u}})} \right\|_{\Ha_{sd}}  {{u}}_+ . \hfill \\ 
\end{gathered} 
\eeqa
It follows that $ {\mathcal{A}}({\mathbf{u}}, t)$ is continuous in both variables. 

When ${f}={0},\; \mathbb{B}$ is the ball of radius $\tf{(1-\e)}{2} u_{+}$. We see from Theorem 9 that this completes the proof of Theorem 3.
\end{proof}

{\bf Proof of Theorem 4}: 
Theorem 3 allows us to conclude that, when ${\mathbf{u}}_0  \in D({\mathbf{A}}) \cap {\mathbb{B}}$, the initial value problem is solved and the solution ${\mathbf{u}}(t,{\mathbf{x}})$ is in $\mathbb{C}^1[(0,\infty);{\mathbb B}]$.  Since $D({\mathbf{A}}) \cap {\mathbb{B}} \subset \mathbb{H}^{2}$, it follows that ${\mathbf{u}}(t,{\mathbf{x}})$ is also in $\mathbb{ V}$ for each $t>0$.  However, we can only conclude  that, for any $T>0$,
\[
\int_0^T {\left\| {{\mathbf{u}}(t)} \right\|_{\mathbb{H}_{sd}}^2 dt}  < \infty ,{\text{ and }}\sup _{0 < t < T} \left\| {{\mathbf{u}}(t)} \right\|_{\mathbb{V}_{sd}}^2  < \infty.
\]
This condition is not strong enough to ensure that $\left\| {{\mathbf{u}}(t)} \right\|_{\mathbb{V}}$ is continuous (which is required to resolve the singularity question). 

{\bf Proof of Theorem 5}:
The proof of Theorem 5 follows from the fact that ${\Ha}^{1/2}$ and ${\bf L}^p, \ 1 \le p \le \iy$ are continuous and densely embedded in ${\bf SD}^2$.  Furthermore, the statement concerning $\mcB$ is obvious.

{\bf Proof of Theorem 6}:
The assertions (1) and (2), follow from Theorem 1.2 in  Brandolese and Vigneron \cite{BV}.

The proof of the first part of  Corollary 7 follows, since every weak (distributional) solution in ${\bf L}^2[\R^3]$ is a strong (norm) solution in ${\bf SD}^2[\R^3]$.  The  second assertion is a special case of a result due to Kato [KA1] (see Remark 1.2 and Theorem 40 ).
\section{The Inhomogeneous Problem} 
In the inhomogeneous case, equation (1) becomes (see \cite{Li2} and \cite{GR}):  
\beqn
\begin{gathered}
  \rho[\partial _t {\mathbf{u}} + ({\mathbf{u}} \cdot \nabla ){\mathbf{u}}] - \mu \Delta {\mathbf{u}} + \nabla p = \rho{\mathbf{f}}(t){\text{ in (}}0,T) \times {\mathbb {R}}^3, \hfill \\
  {\text{                              }}\nabla  \cdot {\mathbf{u}} = 0{\text{ in (}}0,T) \times {\mathbb {R}}^3 {\text{ (in the weak sense),}} \hfill \\
    {\text{                              }}{\mathbf{u}}(0,{\mathbf{x}}) = {\mathbf{u}}_0 ({\mathbf{x}}){\text{ in }}{\mathbb {R}}^3. \hfill \\ 
        {\text{                              }}\frac{{\partial \rho }}
{{\partial t}} + {\mathbf{u}} \cdot \nabla \rho  = 0{\text{ in (}}0,T) \times {\mathbb {R}}^3,  \hfill \\ 
    {\text{                              }}{\rho}(0,{\mathbf{x}}) = {\rho}_0 ({\mathbf{x}}){\text{ in }}{\mathbb {R}}^3. \hfill \\ 
\end{gathered} 
\eeqn
We  assume that the initial density  satisfies $0 \le \rho_0({\bf x}) \le \beta$ for some constant $0<\beta$.  

If we use the Leray projection to eliminate the pressure and divide the first equation by $\rho$, we get 
\beqn
\begin{gathered}
  \partial _t {\mathbf{u}} + \Pa({\mathbf{u}} \cdot \nabla ){\mathbf{u}} - \f{\mu}{\rho} \Pa\Delta {\mathbf{u}}  = \Pa{\mathbf{f}}(t){\text{ in (}}0,T) \times {\mathbb {R}}^3, \hfill \\
    {\text{                              }}{\mathbf{u}}(0,{\mathbf{x}}) = {\mathbf{u}}_0 ({\mathbf{x}}){\text{ in }}{\mathbb {R}}^3. \hfill \\ 
        {\text{                              }}\frac{{\partial \rho }}
{{\partial t}} + {\mathbf{u}} \cdot \nabla \rho  = 0{\text{ in (}}0,T) \times {\mathbb {R}}^3,  \hfill \\ 
    {\text{                              }}{\rho}(0,{\mathbf{x}}) = {\rho}_0 ({\mathbf{x}}){\text{ in }}{\mathbb {R}}^3. \hfill \\ 
\end{gathered} 
\eeqn
The solution of the density equation is well-known, ${\rho}(t,{\mathbf{x}})=U[t,0]\rho_0({\bf x})$, where $U[t,0]$ is an isometry (which depends on ${\bf u}(t, {\bf x})$).  Using the Feynman operator calculus (see \cite{GZ2}), we can write it symbolically as $U[t,0]\rho_0({\bf x})=\rho_0({\bf x}-\int_0^t{{\bf u}(\tau, {\bf x})d\tau})$.  It follows that $0  \le {\rho}(t,{\mathbf{x}}) \le \beta$ for all $t \in [0, T]$.

Equation (21) now becomes:
\beqn
\begin{gathered}
  \left\langle {{\mathcal A}({\mathbf{u}},t),{\mathbf{u}}} \right\rangle _{\Ha_{sd}}  =  - \mu \left\langle { \f{ {\mathbf{Au}} }{\rho},{\mathbf{u}}} \right\rangle _{\Ha_{sd}}  + \left\langle {\left[ { - B({\mathbf{u}},{\mathbf{u}}) + \mathbb{P}{\mathbf{f}}} \right],{\mathbf{u}}} \right\rangle _{\Ha_{sd}}  \hfill \\
   \le  - \f{\mu}{\beta} \left\| {{\mathbf{A}} {\mathbf{u}}} \right\|_{\Ha_{sd}}^2  - \left\langle {B({\mathbf{u}},{\mathbf{u}}),{\mathbf{u}}} \right\rangle _{\Ha_{sd}}  + \left\langle {\mathbb{P}{\mathbf{f}},{\mathbf{u}}} \right\rangle _{\Ha_{sd}}.  \hfill \\ 
\end{gathered} 
\eeqn
It follows that, on setting $\nu=\tf{\mu}{\beta}$, we see that Theorems 4 and 5 also hold for the inhomogeneous problem with minor adjustments.  
\subsection{Bounded Domains}
Finally, a close review of the results of this paper show that they also hold for both homogeneous and inhomogeneous flows on bounded domains, with only minor changes. 

\section{Relationship to Other Spaces}
As noted earlier, Kato and fujita \cite{KF} proved that the Navier-Stokes equations have global solutions for small initial data in ${\dot H}^{1/2}$.  If $\hat h({\mathbf{x}}) = \int_{{\mathbb{R}^3}} {{e^{ - 2\pi i{\mathbf{x}} \cdot {\mathbf{y}}}}h({\mathbf{y}})d{\mathbf{y}}}$ is the Fourier transform of a function $h \in L^2[{\mathbb{R}^3}]$, then  
$h \in {\dot H}^{1/2}[{\mathbb{R}^3}]$ if and only if 
\[
{\left\| h \right\|_{{{\dot H}^{{1}/{2}}}}} = {\left[ {\int_{{\mathbb{R}^3}} {2\pi \left| {\mathbf{x}} \right|{{\left| {\hat h({\mathbf{x}})} \right|}^2}d{\mathbf{x}}} } \right]^{1/2}} < \infty .
\]

This is one of the spaces where the solutions (with zero body forces) are  invariant  under translation and  scaling transformations. If ${\mathbf{u}} (t,{\mathbf{x}})$ is a solution then  ${\mathbf{u}}_\lambda  (t,{\mathbf{x}}) = \lambda {\mathbf{u}}(\lambda ^2 t,\lambda {\mathbf{x}})$ is also one and ${\mathbf{u}}_{0, \la} ({\mathbf{x}}) = \lambda{\bf u}_0({\bf \lambda x})$. (Kato proved a similar result in \cite{KA1}.)

The paper by Koch and Tataru \cite{KT} is of special interest.   They define the smallness of their solution in terms of the norm in  the space $BMO^{-1}$. They consider this space natural for the problem, since it is an invariant space for translations and the scaling transformations,  ${\mathbf{u}}_\lambda  (t,{\mathbf{x}}) = \lambda {\mathbf{u}}(\lambda ^2 t,\lambda {\mathbf{x}})$, that leave the Navier-Stokes equations invariant.

In their approach, Koch and Tataru \cite{KT} use the subspace $BMO^{-1}[\R^n]$  of $BMO[\R^n]$, the functions of bounded mean oscillation to construct strong solutions for the Navier-Stokes equations.  In this section we study the relationship of their work to ours.  (The main result is that  $BMO^{-1}[\R^n] \subset {\bf{SD}}^{2}[\R^n]$.)

In order to do this, we first need to construct ${\bf{SD}}^{p}[\R^n]$ for all $p$. For $f \in \mathbf{L}^p[{\mathbb{R}}^n ]$, define:
\[
\left\| f \right\|_{{\mathbf{SD}}^p }  = \left\{ {\begin{array}{*{20}c}
   {\left\{ {\sum\nolimits_{k = 1}^\infty  {t_k \left| {\int_{\mathbb{R}^n } { {\mathcal{E}}_k ({\mathbf{x}}) \cdot f({\mathbf{x}})d{\mathbf{x}}} } \right|} ^p } \right\}^{1/p} ,1 \leqslant p < \infty},  \\
   {\sup _{k \geqslant 1} \left| {\int_{\mathbb{R}^n } { {\mathcal{E}}_k ({\mathbf{x}}) \cdot f({\mathbf{x}})d{\mathbf{x}}} } \right|,p = \infty .}  \\

 \end{array} } \right.
\] 
It is easy to see that $\left\| \cdot \right\|_{{\mathbf{SD}}^p }$ defines a norm on  $\mathbf{L}^p[{\mathbb{R}}^n ]$.  If ${{\mathbf{SD}}^p[{\mathbb{R}}^n ] }$ is the completion of $\mathbf{L}^p[{\mathbb{R}}^n ]$ with respect to this norm, we have:
\begin{thm} For each $q,\;1 \leqslant q \leqslant \infty,$
 $\mathbf{SD}^p [{\mathbb{R}}^n ] \supset {\mathbf{L}}^q [{\mathbb{R}}^n ]$ as continuous, dense and compact embeddings.
\end{thm}
The proof of this result as well as the next  is essentially the same as in Section 3 of \cite{GZ1}, so we omit them.
\begin{thm} For $\mathbf{SD}^p[{\mathbb{R}}^n ]$, $1\leq p \leq \infty$, we have: 
\begin{enumerate}
\item If $f,g \in \mathbf{SD}^p[{\mathbb{R}}^n ]$, then 
$
\left\| {f + g} \right\|_{{\mathbf{SD}}^{p} }  \leqslant \left\| f \right\|_{{\mathbf{SD}}^{p} }  + \left\| g \right\|_{{\mathbf{SD}}^{p} }
$  (Minkowski inequality). 
\item If $1< p < \infty$ and $p^{ - 1}  + q^{ - 1}  = 1$, then the dual space of $\mathbf{SD}^p[{\mathbb{R}}^n ]$ is $\mathbf{SD}^q[{\mathbb{R}}^n ]$.
\item The dual space of $\mathbf{SD}^1[{\mathbb{R}}^n ]$ is $\mathbf{SD}^{\iy}[{\mathbb{R}}^n ]$.
\item  $ \mathbf{SD}^{\infty}[{\mathbb{R}}^n ] \subset \mathbf{SD}^p[{\mathbb{R}}^n ]$, as a continuous embedding, for $1\leq p < \infty$.
\end{enumerate}
\end{thm}
\section{Functions of Bounded Mean Oscillation}
If $Q$ is a cube in $\R^3$ and $g$ is locally ${\bf L}^1[{\mathbb{R}}^3]$, define $m_{g,Q}$ by:
\[
m_{g,Q}  = \frac{1}
{{\left| Q \right|}}\int_Q { {g({\mathbf{x}})}d{\mathbf{x}}}.
\]
We call it the average or mean of $g$ over $Q$.
\begin{Def} If  $g$ is locally ${\bf L}^1[{\mathbb{R}}^3]$, we say that $g$ is of bounded mean oscillation, and write $g \in BMO[{\mathbb{R}}^3]$, provided that: 
\[
\left\| g \right\|_{BMO}  = \mathop {\sup }\limits_Q \frac{1}
{{\left| Q \right|}}\int_Q {\left| {g({\mathbf{x}}) - m_{g,Q} } \right|d{\mathbf{x}}}  < \infty ,
\]
where the $sup$ is over all cubes in $\R^3$ (see Grafakos \cite{GRA}, chapter 7).
\end{Def}
It is well-known that ${\bf L}^{\iy}[{\mathbb{R}}^3] \subset BMO[{\mathbb{R}}^3]$, with the inclusion proper. For example, $BMO[{\mathbb{R}}^3]$  also contains $ln\lt| p({\bf x}) \rt|$, for all polynomials $p({\bf x})$. 

$BMO^{-1}[\R^n]$ is a subspace of tempered distributions and the following theorem provides a nice characterization. (For a proof, see  Koch and Tataru \cite{KT}.)
\begin{thm} Let the vector field $u$ be a  tempered distribution.  Then $ u \in BMO^{-1}[\R^n]$ if and only if there exist $f^i \in BMO[\R^n]$, such that, ${ u} = \sum\nolimits_{i = 1}^n {{\partial _i}{f^i}}$.
\end{thm}
The following theorem shows that  $BMO^{-1}[\R^n] \subset {\mathbf{SD}}^\iy[\R^n]$, so that necessarily,  $BMO^{-1}[\R^n] \subset {\mathbf{SD}}^2[\R^n]$.
\begin{thm} If ${\bf u} \in {\mathcal S}'[\R^n]$, the space of tempered distributions, then ${\bf u} \in \mathbf{SD}^{\iy}[\R^3]$.
\end{thm}
\begin{proof} Since ${\mathcal{E}}_k ({\mathbf{x}}) \in \C_c^{\iy}[{\mathbb{R}^n }]$ and of slow growth at infinity, it is in ${\mathcal S}[\R^n]$, so that,  
\[
\sup_k \left| {\left\langle {{\mathbf{u}}, \mcE_k } \right\rangle } \right| = \sup_k\left| {\int_{\mathbb{R}^n } {{\mathbf{u}}({\mathbf{x}}) \cdot \mcE_k ({\mathbf{x}})d{\mathbf{x}}} } \right| < \infty ,\quad k \in \mathbb{N}.
\]
It follows that $\bf{u} \in {\mathbf{SD}}^\iy[\R^n]$.
\end{proof}
\subsection{Conclusion}
In this paper we have introduced a new Hilbert space, which allows us to obtain uniqueness for the Leray-Hopf solutions on ${\mathbb R}^3$, with or without body forces.
We also prove global-in-time strong solutions for the three-dimensional Navier-Stokes equations for both bounded and unbounded domains and for a homogeneous or inhomogeneous incompressible fluid.     In addition, with mild conditions on the decay properties of the initial data, we obtain pointwise and time-decay of the solutions.  However, our methods do not allow us to resolve the singularity question.    Our space also contains the Kato solution and  those in ${\bf L}^p$ spaces.  Although the space used by Koch-Tataru \cite{KT}, $BMO^{-1} \subset {\bf{SD}}^2$, we are unable to ensure that the embedding is continuous.  Thus,  we are not able to show that solutions in their sense are solutions in ${\bf{SD}}^2$.

This paper replaces an earlier one, which contained a fatal error that could not be fixed in the manner we had hoped (see \cite{GZ3}).

\end{document}